\documentclass[a4paper,UKenglish,autoref]{lipics-v2021}
\usepackage{xspace}
\usepackage{cite}

\nolinenumbers

\hypersetup{colorlinks,linkcolor=blue,urlcolor=blue,citecolor=blue}
\DeclareFontShape{T1}{lmr}{m}{scit}{<->ssub*lmr/m/scsl}{}
\hideLIPIcs

\newcommand{\OPT}{\textsc{Opt}\xspace}
\newcommand{\ALG}{\textsc{Alg}\xspace}
\newcommand{\Mimic}{\textsc{Mimic}\xspace}

\newcommand{\Disc}{\textsc{Disc}\xspace}

\newcommand{\cost}{\textsc{cost}}
\newcommand{\val}{\textsc{val}}
\newcommand{\I}{\ensuremath{\mathcal{I}}\xspace}
\newcommand{\goodphases}{T}
\newcommand{\E}{\mathbf{E}}
\newcommand{\fresh}{R^{\mathrm{F}}}
\newcommand{\stale}{R^{\mathrm{S}}}
\newcommand{\previous}[1]{P(#1)}
\newcommand{\alpham}{\delta\xspace}
\newcommand{\ws}{w^\mathrm{S}}
\newcommand{\wf}{w^\mathrm{F}}
\newcommand{\g}[2]{\textsc{cost}_{#1}(#2)\xspace}
\newcommand{\gf}{g^\mathrm{F}}
\newcommand{\gs}{g^\mathrm{S}}
\newcommand{\ct}[2]{s_{#1}(#2)\xspace} 
\newcommand{\Dxi}[1]{\xi_{#1}}
\newcommand{\DC}[1]{B_{#1}}
\newcommand{\DE}[1]{C_{#1}}
\newcommand{\DD}[2]{D_{#1 #2}}
\newcommand{\DB}[2]{F_{#1 #2}}
\newcommand{\DG}[2]{G_{#1 #2}}
\newcommand{\DHH}[2]{H_{#1 #2}}
\newcommand{\DRP}[2]{U_{#1 #2}}
\newcommand{\DST}[2]{V_{#1 #2}}

\title{Traveling Repairperson, Unrelated Machines, and Other Stories About Average Completion Times}
\titlerunning{TRP, Unrelated Machines, and Other Stories About Average Completion Times}

\author{Marcin Bienkowski}{Institute of Computer Science, University of Wroc{\l}aw, Poland}{marcin.bienkowski@cs.uni.wroc.pl}{https://orcid.org/0000-0002-2453-7772}{}
\author{Artur Kraska}{Institute of Computer Science, University of Wroc{\l}aw, Poland}{artur.kraska@cs.uni.wroc.pl}{https://orcid.org/0000-0003-0973-787X}{}
\author{Hsiang-Hsuan Liu}{Utrecht University, Netherlands}{h.h.liu@uu.nl}{https://orcid.org/0000-0002-0194-9360}{}

\authorrunning{M. Bienkowski, A. Kraska and H.-H. Liu}
\Copyright{Marcin Bienkowski, Artur Kraska and Hsiang-Hsuan Liu}

\funding{Supported by Polish National Science Centre grant 2016/22/E/ST6/00499.}

\category{Track A: Algorithms, Complexity and Games}



\EventEditors{Nikhil Bansal, Emanuela Merelli, and James Worrell}
\EventNoEds{3}
\EventLongTitle{48th International Colloquium on Automata, Languages, and Programming (ICALP 2021)}
\EventShortTitle{ICALP 2021}
\EventAcronym{ICALP}
\EventYear{2021}
\EventDate{July 12--16, 2021}
\EventLocation{Glasgow, Scotland (Virtual Conference)}
\EventLogo{}
\SeriesVolume{198}
\ArticleNo{26}

\begin{document}

\maketitle

\begin{abstract}
We present a unified framework for minimizing average completion time for many
seemingly disparate \emph{online} scheduling problems, such as the traveling
repairperson problems (TRP), dial-a-ride problems (DARP), and scheduling on
unrelated machines. 

We construct a simple algorithm that handles all these scheduling
problems, by computing and later executing auxiliary schedules, each optimizing
a certain function on already seen prefix of the input. The optimized function
resembles a prize-collecting variant of the original scheduling problem. By a
careful analysis of the interplay between these auxiliary schedules, and later
employing the resulting inequalities in a factor-revealing linear program, we
obtain improved bounds on the competitive ratio for all these scheduling
problems. 

In particular, our techniques yield a $4$-competitive deterministic algorithm
for all previously studied variants of online TRP and DARP, and a $3$-competitive
one for the scheduling on unrelated machines (also with precedence constraints).
This improves over currently best ratios for these problems that are $5.14$
and~$4$, respectively. We also show how to use randomization to further
reduce the competitive ratios to $1+2/\ln 3 < 2.821$ and $1+1/\ln 2 < 2.443$,
respectively. The randomized bounds also substantially improve the current state
of the art. Our upper bound for DARP contradicts the lower bound of 3 given 
by Fink et al.~(Inf.~Process.~Lett.~2009); we pinpoint a flaw in their proof.
\end{abstract}


\section{Introduction}

In the traveling repairperson problem (TRP)~\cite{SahGon76}, requests arrive in
time at points of a~metric space and they need to be eventually serviced. In the
same metric, there is a mobile server, that can move at a constant speed. The
server starts at a distinguished point called the origin. A request is considered
serviced once the server reaches its location; we call such time its
\emph{completion time}. The goal is to minimize the sum (or equivalently the
average) of all completion times. We focus on a weighted variant, where all
requests have non-negative weights and the goal is to minimize the weighted sum
of completion times. 

A natural and well-studied extension of the TRP problem is a so-called
\emph{dial-a-ride problem} (DARP) \cite{PaLSSS04}, where each request has a
source and a destination and the goal is to transport an~object between these
two points. There, the server may have a fixed capacity limiting the number of
objects it may carry simultaneously; this capacity may be also infinite. For the
finite-capacity case, one can also distinguish between preemptive variant, where
objects can be unloaded at some points of the metric space (different than
their destination) and non-preemptive variant, where such unloading is not
allowed.

A seemingly disparate problem is scheduling on $m$ unrelated machines
\cite{GrLaLK79}. There, weighted jobs arrive in time, each with a vector of
size $m$ describing execution times of the job when assigned to a given
machine. A single machine can execute at most one job at a~time. The goal
is to assign each job (at or after its arrival) to one of the machines to
minimize the weighted sum of completion times. This problem comes in two
flavors: in the preemptive one, job execution may be interrupted and picked up
later, while in the non-preemptive one, such interruption is not possible. As
an extension, each job may have precedence constraints, i.e., can be executed
only once some other jobs are completed.

\subparagraph{Online Algorithms.}

Our focus is on natural \emph{online} scenarios of TRP, DARP~\cite{FeuSto01},
and machine scheduling~\cite{HaScSW97}. There, an online algorithm \ALG, at
time~$t$, knows only requests/jobs that arrived before or at time $t$. The
number of requests/jobs is also not known by an~algorithm a~priori. We say that
an online algorithm \ALG is $c$-competitive if for any request/job sequence~$\I$ it
holds that $\cost_{\ALG}(\I) \leq c \cdot \cost_{\OPT}(\I)$, where $\OPT$ is a
cost-optimal \emph{offline} solution for~$\I$. For a~randomized algorithm
$\ALG$, we replace its cost by its expectation. The competitive ratio of \ALG is
the infimum over all values $c$ such that \ALG is
$c$-competitive~\cite{BorEl-98}.
 
In this paper, we present a unified framework for handling such online scheduling
problems where the cost is the weighted sum of completion times. We
present an algorithm \Mimic that yields substantially improved
competitive ratios for all the problems described above.


\subsection{Previous Work}

The currently best algorithms for the TRP, the DARP, and machine scheduling on
unrelated machines share a common framework. Namely, each of these algorithms
works in phases of geometrically increasing lengths. In each phase, it computes
and executes an auxiliary schedule for the requests presented so far. (In the
case of the TRP and DARP, the server additionally returns to the origin
afterward.) The auxiliary schedule optimizes a certain function, such as
maximizing the weight of served
requests~\cite{KrPaPS03,KrPaPS06,JaiWag06,BieLiu19,HaScSW97,CPSSSW96} or
minimizing the sum of completion times with an additional penalty for non-served
requests~\cite{HwaJai18}.\footnote{Computing such auxiliary schedule usually
involves optimally solving an NP-hard task. This is typical for the area of
online algorithms, where the focus is on~information-theoretic aspects and not
on computational complexity. Algorithms presented in this paper also aim at
minimizing the achievable competitive ratio rather than minimizing the running
time.} Moreover, known randomized algorithms are also based on a common idea: they
delay the execution of the deterministic algorithm by a random
offset~\cite{KrPaPS03,KrPaPS06,CPSSSW96,HwaJai18}. We call these approaches
\emph{phase based}. The currently best results are gathered in
\autoref{tab:results}.

\subparagraph*{Traveling Repairperson and Dial-a-Ride Problems.}

The online variant of the TRP has been first investigated by Feuerstein
and Stougie~\cite{FeuSto01}. By adapting an algorithm for the cow-path problem
problem~\cite{BaCuRa93}, they gave a 9-competitive solution for line metrics. 
The result has
been improved by Krumke~et~al.~\cite{KrPaPS03}, who gave a
phase-based deterministic algorithm \textsc{Interval} attaining competitive
ratio of $3 + 2\sqrt{2} < 5.829$ for an arbitrary metric space. A
slightly different algorithm with the same competitive ratio was given by
Jaillet and Wagner~\cite{JaiWag06}. Bienkowski and
Liu~\cite{BieLiu19} applied postprocessing to auxiliary schedules, serving
heavier requests earlier, and improved the ratio to
$5.429$ on line metrics. Finally, Hwang and Jaillet proposed a phase-based
algorithm $\textsc{Plan-And-Commit}$~\cite{HwaJai18}. They give a computer-based
upper bound of $5.14$ for the competitive ratio and an analytical upper bound of
$5.572$.

Randomized counterparts of algorithms \textsc{Interval} and \textsc{Plan-And-Commit}  
achieve ratios of $3.874$~\cite{KrPaPS03,KrPaPS06} and $3.641$~\cite{HwaJai18},
respectively. Interestingly, the latter bound is not a direct randomization of
the deterministic algorithm, but uses a different parameterization, putting more
emphasis on penalizing requests not served by auxiliary schedules.

The phase-based algorithm \textsc{Interval} extends in a straightforward fashion
to the DARP problem with an arbitrary assumption on the server capacity, both
for the preemptive and non-preemptive variants: all the details of the solved
problem are encapsulated in the computations of auxiliary
schedules~\cite{KrPaPS03}. In the same manner, \textsc{Interval} can be enhanced
to handle $k$-TRP and $k$-DARP variants, where an algorithm has $k$ servers at
its disposal (also for any $k$, any server capacities, and any preemptiveness
assumptions)~\cite{BonSto09}. Although this was not explicitly stated
in~\cite{HwaJai18}, the algorithm \textsc{Plan-And-Commit} can be extended in
the same way.

From the impossibility side, Feuerstein and Stougie~\cite{FeuSto01} gave a~lower
bound for the TRP (that also holds already for a line) of $1 + \sqrt{2} > 2.414$,
while the bound of $7/3$ for randomized algorithms was presented by
Krumke~et~al.~\cite{KrPaPS03}. For the variant of the TRP with multiple servers,
the deterministic lower bound is only $2$~\cite{BonSto09} (it holds for any
number of servers). Clearly, all these lower bounds hold also for any variant of
DARP. For the DARP with a~single server of capacity $1$, the deterministic lower
bound can be improved to $3$~\cite{FeuSto01} and the randomized one to
$2.410$~\cite{KrPaPS03}.

The authors of~\cite{FiKrWe09} claimed a lower bound of $3$ for randomized $k$-DARP
(for any $k$). This contradicts the upper bound we present in this paper. In
\autoref{sec:flaw}, we pinpoint a flaw in their argument.

\subparagraph*{TRP and DARP: Related Results.}

Both online TRP and DARP problems were considered under different objectives,
such as minimizing the total makespan (when the TRP becomes online
TSP)~\cite{AsKrRa00,AuFLST94,AuFLST95,AuFLST01,BirDis20,BiDiSc19,BDHHLM17,BlKrPS01,ChDeLW19,JaMuSr19,JaiWag08,LiLPSS04}
or maximum flow time~\cite{HaKrRa00,KPPLMS05,KLLMPP02}.

The offline variants of TRP and DARP have been extensively studied both from the
computational hardness~(see, e.g., \cite{SahGon76,PaLSSS04}) and approximation
algorithms perspectives. In particular, the TRP, also known as the \emph{minimum
latency problem} problem, is NP-hard already on weighted trees~\cite{Sitter02}
(where the closely related traveling salesperson problem~\cite{Blae16} becomes
trivial) and the best known approximation factor in general graphs is
3.59~\cite{ChGoRT03}. For some metrics (Euclidean plane, planar graphs or
weighted trees) the TRP admits a~PTAS~\cite{AroKar03,Sitter14a}.


\begin{table}
\begin{center}
\begin{tabular}{r|l|l|l|l|}
	& \multicolumn{2}{c|}{deterministic} & \multicolumn{2}{c|}{randomized} \\
\hline
 & lower & upper & lower & upper \\
\hline
TRP  & $2.414$~\cite{FeuSto01} & $5.14$~\cite{HwaJai18} & $2.333$~\cite{KrPaPS03} & $3.641$~\cite{HwaJai18} \\
DARP & $3$~\cite{FeuSto01} & $5.14^*$~\cite{HwaJai18} & $2.410$~\cite{KrPaPS03} & $3.641^*$~\cite{HwaJai18} \\
$k$-TRP & $2$~\cite{BonSto09} & $5.14^*$~\cite{HwaJai18} & $2$~\cite{BonSto09} & $3.641^*$~\cite{HwaJai18} \\
$k$-DARP & $2$~\cite{BonSto09} & $5.14^*$~\cite{HwaJai18} & $2$~\cite{BonSto09} & $3.641^*$~\cite{HwaJai18} \\
$k$-TRP, $k$-DARP (all variants) & & $\mathbf{4}$                      &                         & $\mathbf{2.821}$ \\
\hline
\multirow{2}{*}{scheduling on unrelated machines} 
     & \multirow{2}{*}{1.309~\cite{Vestje97}} & $4$~\cite{HaScSW97} 
          & \multirow{2}{*}{$1.157$~\cite{Seiden00}} & $2.886$~\cite{CPSSSW96} \\
     &  & $\mathbf{3}$ & & $\mathbf{2.443}$\\
\hline
\end{tabular}
\end{center}
\caption{Previous and current bounds on the competitive ratios for the TRP and
the DARP problems. Asterisked results were not given in the referenced papers,
but they are immediate consequences of the arguments therein. All upper bounds
for the TRP/DARP variants hold for any number $k$ of servers, any server capacities,
both in the preemptive and the non-preemptive case. Upper bounds for scheduling hold
also in the presence of precedence constraints. Bounds proven in the current
paper are given in boldface.}
\label{tab:results}
\end{table}


\subparagraph{Machine Scheduling on Unrelated Machines.}

The first online algorithm for the scheduling on unrelated machines ($R|r_j|\sum
w_j C_j$ in the Graham et al.~notation~\cite{GrLaLK79}) was given by
Hall~et~al.~\cite{HaScSW97}. They gave 8-competitive polynomial-time algorithm,
which would be $4$-competitive if the polynomial-time requirement was lifted.
Chakrabarti et al.~showed how to randomize this algorithm, achieving the ratio
of $2/\ln 2 < 2.886$~\cite{CPSSSW96}. They also observe that both algorithms
can handle precedence constraints. The currently best deterministic lower of
1.309 is due to Vestjens~\cite{Vestje97}, and the best randomized one of 1.157
is due to Seiden~\cite{Seiden00}.

\subparagraph{Machine Scheduling: Related Results.}

While for unrelated machines, the results have not been beaten for the last 25
years, the competitive ratios for simpler models were improved substantially.
For example, for parallel identical machines, a~sequence of papers lowered the
ratio to $1.791$~\cite{CorWag09,SchSku02,MegSch04,Sitter10}.

The problem has also been studied intensively in the offline regime. Both weighted
preemptive and non-preemptive variants were shown to be
APX-hard~\cite{HoScWo01,Sitter17}. On the positive side, a $1.698$-approximation
for the preemptive case was given by Sitters~\cite{Sitter17}, and a 
$1.5$-approximation for the non-preemptive case by Skutella~\cite{Skutel98}. A
PTAS for a constant number of machines is due to Afrati et al.~\cite{ABCKKK99}.


\subsection{Resettable Scheduling}
\label{sec:resettable}

The phase-based algorithms for DARP variants and machine scheduling on unrelated
machines both execute auxiliary schedules, but the ones for the DARP variants
need to bring the server back to the origin between schedules. We call the
latter action \emph{resetting}. To provide a~single algorithm for all these
scheduling variants, we define a class of \emph{resettable scheduling} problems.

We assume that jobs are handled by an \emph{executor}, which has a set of possible
states. And at time $0$, it is in a distinguished \emph{initial state}. An input
to the problem consists of a~sequence of jobs~$\I$ released over time. Each job
$r$ is characterized by its arrival time $a(r)$, its weight~$w(r)$, and
possibly other parameters that determine its execution time. The executor cannot
start executing job $r$ before its arrival time $a(r)$. We will slightly abuse
the notation and use $\I$ to also denote the \emph{set} of all jobs from the
input sequence. There is a~problem-specific way of executing jobs and we use
$\ct{\ALG}{r}$ to denote the \emph{completion time} of a job by an algorithm
\ALG. The cost of an algorithm is defined as the weighted sum of job completion
times, $\g{\ALG}{\I} = \sum_{r \in \I} w(r) \cdot \ct{\ALG}{r}$. 

For any time $\tau$, let $\I_\tau$ be the set of jobs that appear till
$\tau$. An \emph{auxiliary $\tau$-schedule} is a problem-specific way of
feasibly executing a subset of jobs from $\I_\tau$. Such schedule starts at time
$0$, terminates at time $\tau$, and leaves no job partially executed. We require
that the following properties hold for any resettable scheduling problem. 
\begin{description}
\item[Delayed execution.] At any time $t$, if the executor is in the initial state,
it can execute an~arbitrary auxiliary $\tau$-schedule (for $\tau \leq t$). Such
action takes place in time interval $[t,t+\tau)$. Any job $r$ that would be
completed at time $z \in [0,\tau)$ by the $\tau$-schedule started at time $0$ is
now completed exactly at time $t + z$ (unless it has been already executed
before).

\item[Resetting executor.] Assume that at time $t$, the executor was in the
initial state, and then executed a $\tau$-schedule, ending at time $t+\tau$.
Then, it is possible to \emph{reset} the executor using extra $\gamma \cdot
\tau$ time, where $\gamma$ is a parameter characteristic to the problem. That is,
at time $t+(1+\gamma) \cdot \tau$, the executor is again in  its initial state.

\item[Learning minimum.] We define $\min(\I)$ to be the earliest time at which
\OPT may complete some job. 
We require that the value of $\min(\I)$ is learned by an online
algorithm at or before time $\min(\I)$ and that $\min(\I) > 0$.

\end{description}
We call scheduling problems that obey these restrictions 
$\gamma$-resettable.

\subparagraph{Example 1: Machine Scheduling is 0-Resettable.}

For the machine scheduling problem, the executor is always in the initial state,
and no resetting is necessary. As we may assume that processing of any job
takes positive time, $\min(\I) > 0$ holds for any input $\I$.

\subparagraph{Example 2  : DARP Problems are 1-Resettable.}

For the DARP variants, the executor state is the position of the algorithm
server, with the origin used as the initial state.\footnote{In the variants with
$k$ servers, the executor state is a $k$-tuple describing the positions of all
servers.} Jobs are requests for transporting objects and an auxiliary
$\tau$-schedule is a fixed path of length $\tau$ starting at the origin,
augmented with actions of picking up and dropping particular objects.\footnote{In
the preemptive variants, preemption is allowed \emph{inside} an auxiliary
schedule, provided that after a~$\tau$-schedule terminates, each job is either
completed or untouched.} It is feasible to execute a~$\tau$-schedule starting at
any time $t$ when the server is at the origin. In such case, jobs are completed
with an extra delay of $t$. Furthermore, right after serving the
$\tau$-schedule, the distance between the server and the origin is at
most~$\tau$. Thus, it is possible to reset the executor to the initial state
within extra time $1 \cdot \tau$. 

Finally, as we may assume that there are no requests that arrive at time
$0$ with both start and destination at the origin, $\min(\I) > 0$
for any input $\I$.


\subsection{Our Contribution}

In this paper, we provide a deterministic routine \Mimic and its randomized
version that solves any $\gamma$-resettable scheduling problem. It achieves
a deterministic ratio of $3+\gamma$ and a~randomized one of
$1+(1+\gamma)/\ln(2+\gamma)$.

That is, for $1$-resettable scheduling problems (the DARP variants with
arbitrary server capacity, an arbitrary number of servers, and both in the
preemptive and non-preemptive setting, or the TRP problem with an arbitrary
number of servers), this gives solutions whose ratios are at most $4$ and
$1+2/\ln 3 < 2.821$, respectively. For $0$-resettable scheduling problems (that
include scheduling on unrelated machines with or without precedence
constraints), the ratios of our solutions are $3$ and $1+1/\ln 2 < 2.443$.

In both cases, our results constitute a substantial improvement over currently
best ratios as illustrated in \autoref{tab:results}. Our result for the
scheduling on unrelated machines is the first improvement in the last 25 years
for this problem.

\subparagraph{Challenges and Techniques.}

$\Mimic$ works in phases of geometrically increasing lengths. At
the beginning of each phase, at time $\tau$, it computes an auxiliary
$\tau$-schedule that optimizes the total completion time of jobs seen so far
with an additional penalty for non-completed jobs: they are penalized as if
they were completed at time $\tau$. Then, within the phase it executes this
schedule and afterward it resets the executor. We obtain a~randomized variant
by delaying the start of \Mimic by an offset randomly chosen from a \emph{continuous}
distribution.

Admittedly, this idea is not new, and in fact, when we apply \Mimic to the TRP
problem, it becomes a slightly modified variant of
\textsc{Plan-And-Commit}~\cite{HwaJai18}. Hence, the main technical contribution
of our paper is a careful and exact analysis of such an approach. The crux here is
to observe several structural properties and relations among schedules produced
by \Mimic in consecutive phases, carefully tracking the overlaps of the job
sets completed by them. On this basis, and for a fixed number $Q$ of phases, we
construct a maximization linear program (LP), whose optimal value upper-bounds
the competitive ratio of \Mimic. Roughly speaking, the LP encodes, in a sparse
manner, an adversarially created input. To upper bound its value, we explicitly
construct a solution to its dual (minimization) program and show that its
value is at most~$4$ for any number of phases $Q$. 

Bounding the competitive ratio for the randomized version of \Mimic is
substantially more complicated as we need to combine the discrete world of an LP
with uncountably many random choices of the algorithm. To tackle this issue, we
consider an intermediate solution \Disc which approximates the random choice of
\Mimic to a given precision, choosing an~offset randomly from a discrete set of
$M$ values. This way, we upper-bound the ratio of \Mimic by $1+(1/M) \cdot
\sum_{j=1}^M (2+\gamma)^{j/M}$. This bound holds for an arbitrary value of~$M$,
and thus by taking the limit, we obtain the desired bound on the competitive
ratio. Interestingly, we use the same LP for analyzing both the deterministic and
the randomized solution.


\section{Deterministic and Randomized Algorithms: Routine MIMIC}
\label{sec:algorithm}

To describe our approach for $\gamma$-resettable scheduling, we start with
defining auxiliary schedules used by our routine \Mimic. The parameter $\gamma$
will be used to define partitioning of time into phases. Both our deterministic
and randomized solutions will run \Mimic, however, the randomized one will
execute it for a~random choice of parameters.

\subparagraph{Auxiliary Schedules.}

As introduced already in \autoref{sec:resettable}, an (auxiliary) $\tau$-schedule~$A$ 
describes a sequence of job executions, has the total duration $\tau$, and may be executed
whenever the executor is in the initial state. For the preemptive variants, we assume that 
once such a schedule terminates, each job is processed either completely or not at all.

For a fixed input $\I$, and a $\tau$-schedule $A$, we use $R(A)$ to denote the set of
jobs that would be served by $A$ if it was executed from time $0$, i.e., in
the interval $[0,\tau)$. For any set of jobs $R \subseteq R(A)$, let
\begin{align}
\textstyle    w(R) = \sum_{r \in R} w(r) && \text{and} &&
\textstyle \g{A}{R} = \sum_{r \in R} w(r) \cdot \ct{A}{r} .
\end{align}
Note that if a schedule $A$ serves all jobs from the input ($R(A) = \I$), 
then $\g{A}{R(A)}$ coincides with the cost of an algorithm that executes schedule 
$A$ at time $0$. 

Recall that $\I_\tau \subseteq \I$ denotes the set of jobs that arrive till time $\tau$.
For any $\tau$-schedule $A$, we define its value as 
\begin{equation}
\label{eq:def_val}
    \val_\tau(A) = \g{A}{R(A)} 
	+ \tau \cdot w \left( \I_\tau \setminus R(A) \right).
\end{equation}
The value corresponds to the actual cost of completing jobs from $\I_\tau$ by
schedule $A$ in interval~$[0, \tau)$, but we charge $A$ for unprocessed jobs as if they were completed 
at time $\tau$. 

\begin{definition}
\label{def:s_tau}
For any $\tau \geq 0$, let $S_\tau$ be the $\tau$-schedule minimizing function $\val_\tau$.
Ties are broken arbitrarily, but in a deterministic fashion.
\end{definition}

\subparagraph{Routine MIMIC.}

For solving the $\gamma$-resettable scheduling problem, we define routine
$\Mimic(\gamma,\omega)$, where $\omega \in (-1,0]$ is an additional parameter that controls the initial delay.
\begin{itemize}
\item Our deterministic algorithm is simply $\Mimic(\gamma, 0)$.
\item Our randomized algorithm first chooses a value $\omega$ uniformly at
random from the range~$(-1, 0]$. Then, it executes $\Mimic(\gamma, \omega)$.
\end{itemize}

Internally, $\Mimic(\gamma, \omega)$ uses a parameter $\alpha = 2 + \gamma$. It
splits time into phases in the following way. For any~$k$, let $\tau_k = \tau(k)
= \min(\I) \cdot \alpha^{k+\omega}$. The $k$-th phase (for $k \geq 1$) starts at
time $\tau_{k-1} = \min(\I) \cdot \alpha^{k-1+\omega}$ and ends at time $\tau_{k} =
\min(\I) \cdot \alpha^{k+\omega}$. The time interval $[0,\tau_0) =
[0,\alpha^{\omega} \cdot \min(\I))$ does not belong to any phase. As 
$\alpha^\omega \cdot \min(\I) \leq \min(\I)$, no jobs can be completed
within this interval, by the definition of $\min(\I)$ (see \autoref{sec:resettable}).

\Mimic does nothing till the end of phase~$1$ (till time $\tau_1 =
\alpha^{1+\omega} \cdot \min(\I)$). Since $\omega \geq -1$, we have $\tau_1 \geq
\min(\I)$. As $\Mimic$ learns the value of $\min(\I)$ latest at time $\min(\I)$,
it can thus correctly identify the value of $\tau_1$ before or at time $\tau_1$.

For a phase $k+1$, where $k \geq 1$, \Mimic behaves in the following way. We
ensure that at time $\tau_k$, at the beginning of phase $k+1$, \Mimic is in its
initial state. At this time, \Mimic computes the $\tau_k$-schedule~$S_{\tau(k)}$
(see \autoref{def:s_tau}), executes it within time interval $[\tau_k, 2 \cdot
\tau_k)$ and afterwards, it resets its state to the initial one. The execution
of $S_{\tau(k)}$ will not be interrupted or modified when new jobs arrive within
phase $k+1$. Furthermore, \Mimic serves only those requests from $S_{\tau(k)}$
it has not yet served earlier. The resetting part takes time $\gamma \cdot \tau_k$, and is thus
finished at time $(2+\gamma) \cdot \tau_k = \alpha\cdot \tau_k = \tau_{k+1}$ when
the next phase starts. An illustration is given in \autoref{fig:alg}.


\begin{figure}
\centering
\includegraphics[width=0.9\textwidth]{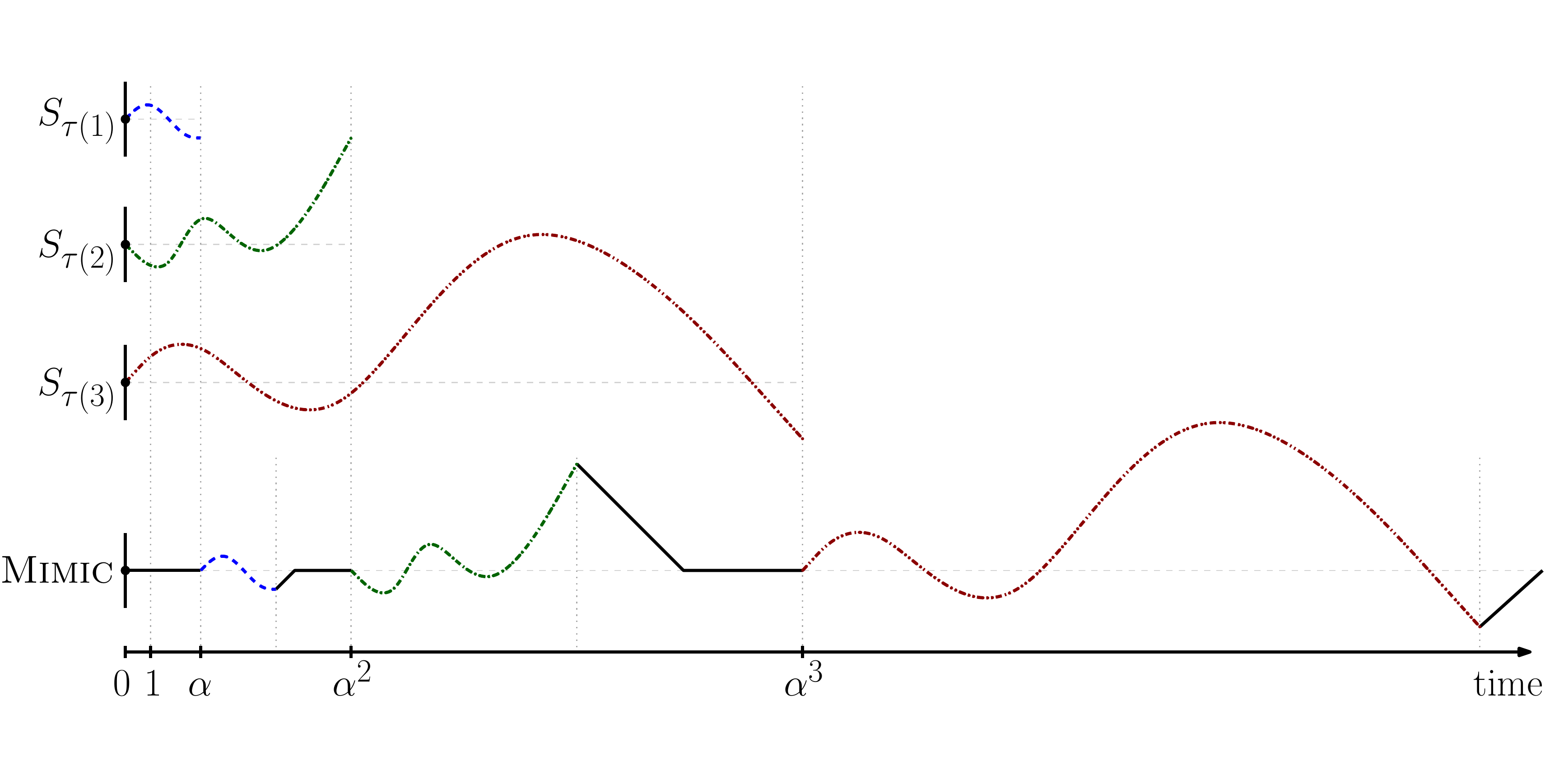}
\caption{An example execution of algorithm $\Mimic(1,0)$ applied for the TRP
problem (i.e., we use $\alpha=3$). We assume that $\min(\I) = 1$. Within time interval $[\tau(k) =
\alpha^k, 2\cdot \alpha^k)$ of phase $k+1$, \Mimic executes a $\tau(k)$-schedule
$S_{\tau(k)}$ that optimizes function $\val_{\tau(k)}$. Afterwards within time
interval $[2 \cdot \alpha^k, \tau(k+1) = 3 \cdot \alpha^k)$, \Mimic resets its
state to the initial one (the server of TRP returns to the origin).}
\label{fig:alg}
\end{figure}


\section{Intermediate Algorithm DISC}

As mentioned in the introduction, we introduce an additional intermediate
algorithm $\Disc$, whose analysis will allow us to bound the competitive ratios
of both our deterministic and randomized solution. For an integer $\ell$, we use
$[\ell]$ to denote the set $\{0,\ldots,\ell-1\}$.

$\Disc(\gamma,M,\beta)$ solves the $\gamma$-resettable scheduling problem, and is
additionally parameterized by a positive integer $M$, and a real number $\beta
\in (0, 1/M]$. 
$\Disc(\gamma, M, \beta)$ first chooses a~random integer $m \in [M]$.
Then, it executes $\Mimic(\gamma, \omega = -1 + m/M + \beta)$.
The main result of this paper is the following bound, whose proof
is will be given in the next two sections.

\begin{theorem}
\label{thm:main}
For any $\gamma$, any positive integer $M$, and any $\beta \in (0,1/M]$, the
competitive ratio of $\Disc(\gamma, M, \beta)$ for the $\gamma$-resettable
scheduling is at most $1 + (1/M) \cdot \sum_{j=1}^M (2+\gamma)^{j/M}$.
\end{theorem}

\begin{corollary}
For any $\gamma$, the competitive ratio of our $\Mimic$-based 
deterministic solution is at most $3+\gamma$
and the ratio of randomized one at most $1+(1+\gamma)/\ln(2+\gamma)$.
\end{corollary}

\begin{proof}
Let $\xi_M = 1 + (1/M) \cdot \sum_{j=1}^M \alpha^{j/M}$. First, we note that
$\Disc(\gamma, M=1, \beta=1)$ chooses deterministically $m = 0$ and executes
$\Mimic(\gamma, \omega=-1+0+1=0)$, i.e., is equivalent to our deterministic
algorithm. Hence, by \autoref{thm:main}, the corresponding competitive ratio is
at most $\xi_1 = 3+\gamma$.

For analyzing our randomized algorithm, we observe that instead of choosing
a~random~$\omega \in (-1,0]$, we may choose a random integer $m \in [M]$ and a
random real $\beta \in (0,1/M]$ and set $\omega = -1 + m/M + \beta$. Thus, for
any fixed integer $M$, our randomized algorithm is equivalent to choosing random
$\beta \in (0,1/M]$ and running~$\Disc(\gamma, M, \beta)$.

Fix any input $\I$. By \autoref{thm:main}, $\E_m[\cost_{\Disc(\gamma, M,
\beta)}(\I)] \leq \xi_M \cdot \cost_{\OPT}(\I)$ holds for any $\beta \in (0,1/M]$,
where the expected value is taken over random choice of $m$. Clearly, this
relation holds also when $\beta$ is chosen randomly, i.e.,
$\E_\omega[\cost_{\Mimic(\gamma,\omega)}] = \E_\gamma \E_m[\cost_{\Disc(\gamma,
M, \beta)}(\I)] \leq \xi_M \cdot \cost_{\OPT}(\I)$. As the bound is valid for any
$M$, and the competitive ratio of our randomized algorithm is at most $\inf_{M
\in \mathbb{N}} \{ \xi_M \} = \lim_{M \to \infty} \xi_M = 1 + (1+\gamma)/
\ln(2+\gamma)$. 
\end{proof}


\section{Structural Properties of DISC}

In this section, we build relations useful for analyzing the performance of 
$\Disc(\gamma, M, \beta)$ on any instance $\I$ of the $\gamma$-resettable scheduling problem.

We start by presenting structural properties of schedules~$S_\tau$. We note
that even if there exists a $\tau$-schedule~$A$ that completes all jobs from $\I$,
$S_\tau$ may leave some jobs untouched. However, a sufficiently
long schedule $S_\tau$ completes all jobs.

\begin{lemma}
\label{lem:goodphases}
Fix any input $\I$. There exists a value $\goodphases_\I$, such that 
for any $\tau \geq \goodphases_\I$, $S_\tau$ completes all jobs of $\I$ and is 
an~optimal (cost-minimal) solution for $\I$.
\end{lemma}

\begin{proof}
Let \OPT be a cost-optimal schedule for $\I$ and let $t$ be its length. Let $w$
be the weight of the lightest job from $\I$. We fix $\goodphases_\I = \max \{ t,
(\val_t(\OPT)+1) / w \}$. Now, we pick any $\tau \geq \goodphases_\I$, and
investigate properties of $S_\tau$.

As $\tau \geq \goodphases_\I \geq t$, the schedule of \OPT can be trivially
extended to a $\tau$-schedule $A$ that does nothing in its suffix of length
$\tau - t$. Both $A$ and \OPT complete all jobs, and thus $\val_\tau(A) =
\val_t(\OPT)$. Moreover, as $S_\tau$ minimizes function $\val_\tau$,
$\val_\tau(S_\tau) \leq \val_\tau(A) = \val_t(\OPT) < \goodphases_\I \cdot w
\leq \tau \cdot w$, and thus $S_\tau$ completes all jobs (as otherwise
$\val_\tau$ would include a penalty of at least $\tau \cdot w$). As $S_\tau$ and
$\OPT$ complete all jobs, $\cost_{S_\tau}(\I) = \val_\tau(S_\tau) \leq
\val_t(\OPT) = \cost_\OPT(\I)$, i.e., $S_\tau$ is an optimal solution for $\I$.
\end{proof}

\subparagraph{Sub-phases.}

Recall that the algorithm $\Disc(\gamma,M,\beta)$ chooses a random integer $m
\in [M]$, and executes $\Mimic(\gamma, \omega=-1+m/M+\beta)$. To compare $\Disc$
executions for different random choices, we introduce sub-phases. 
Recall that $\alpha = 2+\gamma$; let $\alpham = \alpha^{1/M}$. 

Recall that the $k$-th phase of \Mimic starts 
at time $\tau_{k-1}$ and ends at time~$\tau_k$, where 
$\tau_k = \min(\I) \cdot \alpha^{k-1+m/M + \beta} = 
\min(\I) \cdot \alpha^{\beta-1} \cdot \alpham^{m + k \cdot M}$.
For any $q$, we define
\begin{equation}
\label{eq:eta_def}
	\eta_q = \eta(q) = \min(\I) \cdot \alpha^{\beta-1} \cdot \alpham^{q}.
\end{equation}
In these terms, $\tau_k = \eta_{m + k \cdot M}$. We define the $q$-th sub-phase
(for $q \geq 0$) as the time interval starting at time $\eta_{q-1}$ and ending
at time $\eta_{q}$. Then, phase $k$ of $\Disc(\gamma,M,\beta)$ consists of
exactly $M$ sub-phases, numbered from $(k-1) \cdot M + m + 1$ to $k \cdot M +
m$. An example of phases and sub-phases is given in \autoref{fig:phases}. We
emphasize that the start and the end of a~sub-phase is a deterministic function 
of the parameters of $\Disc$, while the start and end of a phase depend
additionally on the value $m \in [M]$ that \Disc chooses randomly.

Recall that our deterministic algorithm is equivalent to $\Mimic(\gamma,0)
\equiv \Disc(\gamma, 1, 1)$. In this case $m = 0$, and thus $\eta_q =
\tau_q$ for any $q$, i.e., each phase consists of one sub-phase, and their
indexes coincide. 

\subparagraph{Sub-phases vs Auxiliary Schedules.}

We now identify the times when auxiliary schedules are computed by
$\Disc(\gamma,M,\beta)$. Recall that at the beginning of any phase 
$k+1$ (where $k \geq 1$), i.e., at time $\tau_k = \eta_{m+k\cdot M}$,
$\Disc$ computes and executes schedule $S_{\eta(m+k\cdot M)}$.
Let $\goodphases_\I$ be the threshold guaranteed by \autoref{lem:goodphases}
and we define $K_\I$~as the smallest integer satisfying $\eta({K_\I \cdot M}) \geq \goodphases_\I$.
Note that $K_\I$ is a deterministic function of input $\I$.

For any choice of $m \in [M]$, the schedule $S_{\eta(m+K_\I \cdot M)}$ completes
all jobs. This schedule is executed by $\Disc$ in phase $K_\I+1$, and thus
$\Disc$ terminates latest in phase $K_\I+1$. Summing up, $\Disc(\gamma,M,\beta)$
executes schedules $S_{\eta(m+M)}, S_{\eta(m+2M)}, \dots, S_{\eta(m+K_\I\cdot M)}$.
At the beginning of the first phase, $\Disc$ does nothing, but for notational
ease, we assume that in the first phase, it also computes and executes a dummy
schedule $S_{\eta(m)}$, which does not complete any job. For succinctness, we
use $A_q = S_{\eta(q)}$. In these terms, $\Disc(\gamma,M,\beta)$ executes
schedules $A_{m+k\cdot M}$ for $k \in [K_\I+1]$. 

Let $Q = K_\I \cdot M + (M-1)$: possible schedule indexes used
by \Disc range from $0$ to $Q$.
For any schedule $A_q$, we define the set of indexes of \emph{preceding schedules} 
$\previous{q} = \{q', q'+M, \ldots, q-M\}$, where $q' = q \mod M$.


\begin{figure}
\centering
\includegraphics[width=0.9\textwidth]{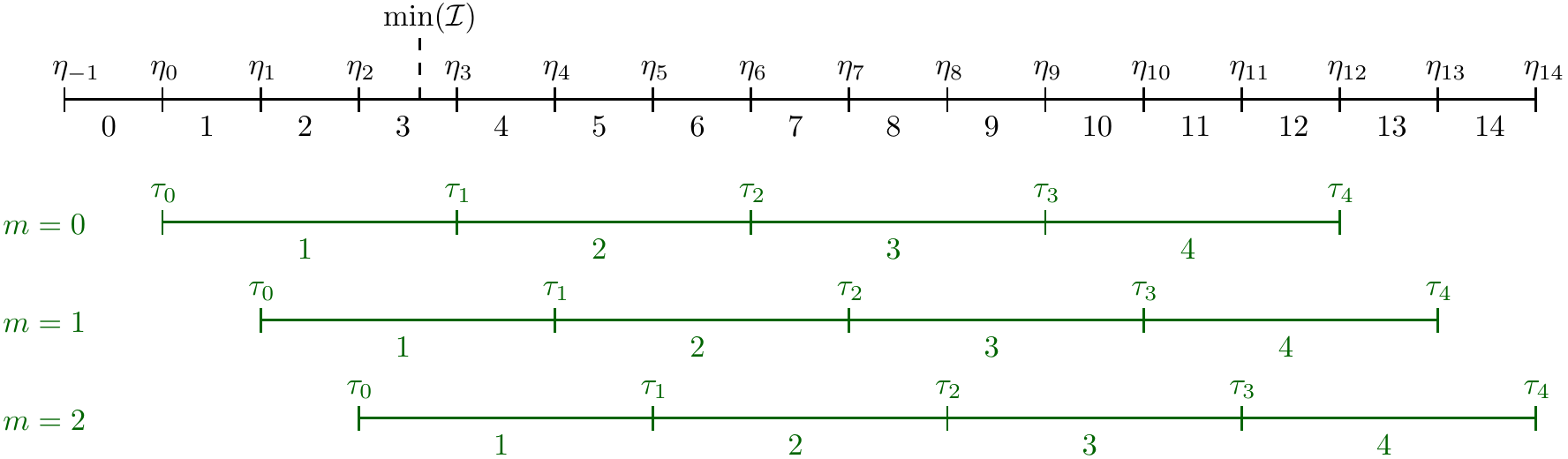}
\caption{Example of phases (green) and sub-phases (black) of algorithm
$\Disc(\gamma,M=3,\beta)$ for all possible choices of $m$. The time interval
lengths are in logarithmic scale. The starts and ends of sub-phases are
deterministic functions of $\gamma$, $M$, and $\beta$, but the start of a~phase
depends additionally on the integer $m \in [M]$ chosen randomly by $\Disc$.
Sub-phase $0$ is not contained in any phase, but will be used in our analysis.}
\label{fig:phases}
\end{figure}


\subparagraph{Fresh and Stale Requests.}

We assume that no jobs are completed by the online algorithm while it is
resetting the executor, and we assume that the execution of schedule $A_q$
may complete only jobs from set~$R(A_q)$.
It is however important to note that $R(A_q)$ and $R(A_{q-M})$ may overlap significantly,
in which case the execution of schedule $A_q$ serves only these jobs from $R(A_q)$ 
that have not been served already. 
To further quantify this effect, for $q \in [Q+1]$,
we define the set of \emph{fresh} jobs of schedule~$A_q$ as 
\begin{equation}
\textstyle	\fresh(A_q) = R(A_q) \setminus \bigcup_{\ell \in \previous{q}} R(A_\ell).
\end{equation}
The remaining jobs from $R(A_q)$ are called \emph{stale} and are denoted
$\stale(A_q) = R(A_q) \setminus \fresh(A_q)$. 
For succinctness, we define the
following shorthand notations for their weights:
\begin{align}
	\wf_q = w(\fresh(A_q)), && \ws_q = w(\stale(A_q)), && 
	w_q = w(R(A_q)) = \wf_q + \ws_q.
\end{align}

\begin{lemma}
\label{lem:weights-relation}
For any $q \in [Q+1]$, it holds that
$\ws_q \leq \sum_{\ell \in \previous{q}} \wf_\ell$. This relation
becomes equality for $q \geq K_\I \cdot M$.
\end{lemma}

\begin{proof}
By a simple induction, it can be shown that $\biguplus_{\ell \in \previous{q}}
\fresh(A_\ell) = \bigcup_{\ell \in \previous{q}} R(A_\ell)$ for any $q \in [Q+1]$.
Then, using the definition of stale jobs, $\stale(A_q) \subseteq
\bigcup_{\ell \in \previous{q}} R(A_\ell) = \biguplus_{\ell \in \previous{q}} \fresh(A_\ell)$.
Applying weight to both sides yields
$\ws_q \leq \sum_{\ell \in \previous{q}} \wf_\ell$.

Next, we show that this relation can be reversed for $q \geq K_\I \cdot M$ (i.e., for the schedule executed 
in the last phase of $\Disc$). For such $q$, $A_q$ completes all jobs, and thus 
$\bigcup_{\ell \in \previous{q}} R(A_\ell)
\subseteq R(A_q) = \fresh(A_q) \uplus \stale(A_q)$. By the definition of fresh
jobs, $\fresh(A_q)$ does not contain any job from
$\bigcup_{\ell \in \previous{q}} R(A_\ell)$, and thus $\bigcup_{\ell \in \previous{q}} 
R(A_\ell)
\subseteq \stale(A_q)$. This implies that $\biguplus_{\ell \in \previous{q}}
\fresh(A_\ell) = \bigcup_{\ell \in \previous{q}} R(A_\ell) \subseteq \stale(A_q)$.
After applying weights to both sides, we obtain
$\ws_q \geq \sum_{\ell \in \previous{q}} \wf_\ell$ as desired.
\end{proof}

\subparagraph{Jobs Completed in Sub-phases.}

For further analysis, we refine our notions when a job is completed. For a
$\eta_q$-schedule~$A_q$, let $R_j(A_q)$ be the set of jobs completed in
sub-phase $j \leq q$, i.e., within interval $[\eta_{j-1},\eta_j)$.
As $\eta_{-1} \leq \eta_{m-1} \leq \min(\I)$ (cf.~\eqref{eq:eta_def}), 
no job can be completed within the interval $[0,\eta_{-1})$ (before sub-phase $0$).
Hence, $R(A_q) = \biguplus_{j=0}^q R_j(A_q)$.

We partition sets $\fresh(A_q)$ and $\stale(A_q)$ analogously, defining sets $\fresh_j(A_q)$ 
and $\stale_j(A_q)$ (for $0 \leq j \leq q$), such that 
$\fresh(A_q) = \biguplus_{j = 0}^q \fresh_j(A_q)$ and
$\stale(A_q) = \biguplus_{j = 0}^q \stale_j(A_q)$.
	For succinctness, 
for $0 \leq j \leq q$, we introduce the following shorthand notations: 
\begin{itemize}
\item $\wf_{qj} = w(\fresh_j(A_q))$, $\ws_{qj} = w(\stale_j(A_q))$, 
	and $w_{qj} = w(R_j(A_q)) = \wf_{qj} + \ws_{qj}$;
\item $\gf_{qj} = \g{A_q}{\fresh_j(A_q)}$, $\gs_{q j} = \g{A_q}{\stale_j(A_q)}$, 
	and $g_{q j} = \g{A_q}{R_j(A_q)} = \gf_{q j} + \gs_{q j}$.
\end{itemize}

\begin{lemma}
\label{lem:two_schedules}
For any $0 \leq q < \ell \leq Q$, it holds that
$\sum_{j=0}^{q} (g_{q j} - g_{\ell j}) + \sum_{j=0}^{q} \eta_q \cdot 
(w_{\ell j} - w_{q j}) \leq 0$.
\end{lemma}

\begin{proof}
For any $\eta_q$-schedule $B$, it holds that
\begin{align*}
\val_{\eta(q)}(B) 
	& = \cost_B(R(B)) + \eta_q \cdot w \left( \I_{\eta(q)} \setminus R(B) \right) \\
	& \textstyle = \sum_{j=0}^q \cost_B(R_j(B)) + \eta_q \cdot w(\I_{\eta(q)}) -
		\eta_q \cdot \sum_{j=0}^q w(R_j(B)).
\end{align*}
Fix any $\ell \leq Q$ and let $A^q_{\ell}$ be the $\eta_q$-schedule consisting
of the first $q$ sub-phases of $\eta_\ell$-schedule~$A_{\ell}$. Since $A_q$ is
a minimizer of $\val_{\eta(q)}$, it holds that $\val_{\eta(q)}(A_q) \leq
\val_{\eta(q)}(A^q_{\ell})$. Thus, $\sum_{j=0}^q g_{q j} -  
		\eta_q \cdot \sum_{j=0}^q w_{q j} 
	\leq 
		\sum_{j=0}^q g_{\ell j} -  
		\eta_q \cdot \sum_{j=0}^q w_{\ell j}.
$
\end{proof}

\subparagraph{Costs of DISC and OPT.}

Finally, we can express costs of $\Disc$ and $\OPT$ using the newly introduced notions.

\begin{lemma}
\label{lem:disc-cost}
For any input $\I$, parameters $M$ and $\beta \in (0,1/M]$, it holds that 
$\E[\cost_\Disc(\I)] $ $= (1/M) \cdot \sum_{q=0}^{Q} \sum_{j=0}^{q} \left( \eta_q \cdot \wf_{qj} + \gf_{q j} \right)$.
\end{lemma}

\begin{proof}
Recall that $\Disc$ chooses random $m \in [M]$ and then at time $\eta_q$ it
executes schedule~$A_q$, for all $q \in \{m, m+M, \dots, m+K_\I \cdot M\}$. When
$\Disc$ executes~$A_q$, it completes jobs from $\fresh(A_q)$. By the delayed
execution property of the resettable scheduling (cf.~\autoref{sec:resettable}),
each job $r \in \fresh(A_q)$ is completed at time $\eta_q + \ct{A_q}{r}$. Thus,
the cost of executing $A_q$ by \Disc is equal to 
\begin{align*}
	\textstyle \sum_{r \in \fresh(A_q)}
	w(r) \cdot (\eta_q  + \ct{A_q}{r}) 
	& = \eta_q \cdot w(\fresh(A_q)) + \g{A_q}{\fresh(A_q)} \\
	& \textstyle = \eta_q \cdot \wf_q + \sum_{j=0}^q \gf_{q j} =
	\sum_{j=0}^q \left( \eta_q \cdot \wf_{qj} + \gf_{q j} \right). 
\end{align*}
For any $q \in
[Q+1]$, the probability that \Disc executes $A_q$ is equal to $1/M$, and thus
the lemma follows.
\end{proof}

\begin{lemma}
\label{lem:opt-cost}
For any input $\I$ 
and any $q \in \{Q-M+1,Q-M+2, \dots, Q\}$, it 
holds that $\cost_\OPT(\I) = \sum_{j=0}^{q} g_{qj}$.
\end{lemma}

\begin{proof}
Recall that for such choice of $q$, schedules $A_q$ 
serve all jobs of $\I$ achieving optimal cost. 
Therefore, $\cost_\OPT(\I) = \g{A_q}{R(A_q)} 
= \sum_{j=0}^{q} \g{A_q}{R_j(A_q)}
= \sum_{j=0}^{q} g_{q j}$.
\end{proof}


\section{Factor-Revealing Linear Program}
\label{sec:lp}

Now we show that the \Disc-to-\OPT cost ratio on an arbitrary input $\I$ 
can be upper-bounded by a value of a linear (maximization) program.

Assume we fixed $\gamma$ and any input $\I$ to the 
$\gamma$-resettable	scheduling problem. We also fix parameters of $\Disc$: an integer 
$M$ and $\beta \in (0,1/M]$. These choices imply the values of~$Q$ and $\eta_q$ for any $q$.
This allows us to define the linear program $\mathcal{P}_{\gamma,\I,M,\beta}$ whose
goal is to maximize
\begin{equation}
\textstyle	\sum_{q=0}^{Q} \sum_{j=0}^{q} \eta_q \cdot \wf_{qj} + \gf_{q j}
\end{equation}
subject to the following constraints:
\begin{align}
	\label{eq:opt-rel}
\textstyle	\sum_{j=0}^{q} g_{q j} \leq 1 
		&&  \text{for all}\; Q-M+1 \leq q \leq Q   \\
	\label{eq:weight-rel-1}
\textstyle	\sum_{j=0}^{q} \ws_{qj} 
		- \sum_{\ell \in \previous{q}} \sum_{j=0}^{\ell} \wf_{\ell j} \leq 0 
		&&  \text{for all}\; 0 \leq q \leq Q-M \\
	\label{eq:weight-rel-2}
\textstyle	\sum_{\ell \in \previous{q}} \sum_{j=0}^{\ell} \wf_{\ell j} -
		\sum_{j=0}^{q} \ws_{qj} \leq 0
		&&  \text{for all}\;  Q-M+1 \leq q \leq Q \\
	\label{eq:weight-rel-mix}
\textstyle	\sum_{j=0}^{q} (g_{q j} - g_{\ell j}) + \sum_{j=0}^{q} \eta_q \cdot (w_{\ell j} - w_{q j}) \leq 0
	  	&&  \text{for all}\; 0 \leq q < \ell \leq Q \\
  \label{eq:weight-rel-last}
	  \eta_{j-1} \cdot \ws_{qj} - \gs_{q j} \leq 0
		  && \text{for all}\; 0 \leq j \leq q \leq Q \\  
  \label{eq:weight-rel-easy-1}
	\gf_{q j} - \eta_j \cdot \wf_{qj} \leq 0
		&& \text{for all}\; 0 \leq j \leq q \leq Q  \\
	\label{eq:weight-rel-easy-2}
	\eta_{j-1} \cdot \wf_{q j} - \gf_{q j} \leq 0 
		&& \text{for all}\; 0 \leq j \leq q \leq Q 
\end{align}
and non-negativity of all variables.
In \eqref{eq:weight-rel-mix}, we treat $w_{qj}$ and $g_{qj}$ not as variables, but 
as shorthand notations for 
$\wf_{qj} + \ws_{qj}$ and $\gf_{qj} + \gs_{qj}$, respectively.

The intuition behind this LP formulation is that instead of creating the whole input $\I$,
the adversary only chooses the values of variables $\wf_{qj}$, $\ws_{qj}$, 
$\gf_{q j}$ and $\gs_{q j}$ that satisfy some subset of inequalities (inequalities that have to be satisfied 
if these variables were created on the basis of actual input $\I$). This intuition is formalized 
below. 

\begin{lemma}
\label{lem:comp-ratio-lp}
Fix any $\gamma$, any input $\I$ for $\gamma$-resettable scheduling,
and parameters of \Disc: integer~$M$ and $\beta \in (0,1/M]$. Then, $\E[\cost_\Disc(\I)] /
\cost_\OPT(\I) \leq P^*_{\gamma,\I,M,\beta} / M$, where
$P^*_{\gamma,\I,M,\beta}$ is the value of the optimal solution to
$\mathcal{P}_{\gamma,\I,M,\beta}$.
\end{lemma}

\begin{proof}
By scaling all variables by the same value, $\mathcal{P}_{\gamma,\I,M,\beta}$ is
equivalent to the (non-linear) optimization program
$\mathcal{P}'_{\gamma,\I,M,\beta}$, whose objective is to maximize
$(\sum_{q=0}^{Q} \sum_{j=0}^{q} \eta_q \cdot \wf_{qj} + \gf_{q j}) / \max_{Q-M+1
\leq q \leq Q} \sum_{j=0}^{q} g_{q j}$, subject to constraints
\eqref{eq:weight-rel-1}--\eqref{eq:weight-rel-easy-2}.
In particular, the optimal values of these programs, $P^*_{\gamma,\I,M,\beta}$ and 
$P'^*_{\gamma,\I,M,\beta}$ are equal.

Next, we set the values of variables $\wf_{qj}$, $\ws_{qj}$, $\gf_{qj}$ and
$\gs_{qj}$ on the basis of input $\I$, and parameters $M$ and $\beta$. (Note
that the variables depend on these parameters, but not on the random choices of
$\Disc$.) We now show that they satisfy the constraints of $\mathcal{P}'^*_{\gamma,\I,M,\beta}$
and we relate $\E[\cost_\Disc(\I)] / \cost_\OPT(\I)$ to 
$P^*_{\gamma,\I,M,\beta}$.

By \autoref{lem:weights-relation} and the relations 
$\wf_{q} = \sum_{j=0}^q \wf_{qj}$ and 
$\ws_{q} = \sum_{j=0}^q \ws_{qj}$, 
the variables satisfy \eqref{eq:weight-rel-1} and \eqref{eq:weight-rel-2}. 
Next, \autoref{lem:two_schedules} implies \eqref{eq:weight-rel-mix}.
Inequalities \eqref{eq:weight-rel-last}, \eqref{eq:weight-rel-easy-1} and \eqref{eq:weight-rel-easy-2}
 follow directly by the definition of costs and weights.
Finally, by \autoref{lem:disc-cost} and \autoref{lem:opt-cost}, for any $q
\in \{Q-M+1, \dots, Q\}$, it holds that $\E[\cost_\Disc(\I)] / \cost_\OPT(\I) =
(1/M) \cdot (\sum_{q=0}^{Q} \sum_{j=0}^{q} \eta_q \cdot \wf_{qj} + \gf_{q j})/
(\sum_{j=0}^{q} g_{q j})$, and thus $\E[\cost_\Disc(\I)] / \cost_\OPT(\I)
\leq P'^*_{\gamma,\I,M,\beta} / M = P^*_{\gamma,\I,M,\beta}/M$.
\end{proof}


\subsection{Dual Program and Competitive Ratio.}

By \autoref{lem:comp-ratio-lp}, the optimal value of
$\mathcal{P}_{\gamma,\I,M,\beta}$ is an~upper bound on the competitive ratio of
\Disc. By weak duality, an upper-bound is given by any feasible
solution to the dual program $\mathcal{D}_{\gamma,\I,M,\beta}$ that we present
below.

$\mathcal{D}_{\gamma,\I,M,\beta}$ uses variables $\Dxi{q}, \DC{q}, \DE{q},
\DD{\ell}{q}, \DB{q}{j}, \DG{q}{j}$, and $\DHH{q}{j}$, corresponding to
inequalities \eqref{eq:opt-rel}--\eqref{eq:weight-rel-easy-2} from $\mathcal{P}_{\gamma,\I,M,\beta}$, 
respectively. In the formulas below, we use $L_q = M \cdot K + (q \mod M)$ and 
$S(q) = \{ q+M, q+2 \cdot M, \dots, L_q-M \}$.
For succinctness of the description, we introduce two 
auxiliary variables for any $0 \leq j \leq q \leq Q$:
\begin{align}
\label{eq:UV_def}
\textstyle \DRP{q}{j} = \sum_{\ell = q+1}^{Q}  \DD{\ell}{q}
			- \sum_{\ell = j}^{q-1}  \DD{q}{\ell}
	&&\text{and}&&
	\textstyle 	\DST{q}{j} = \sum_{\ell = j}^{q-1}  \eta_\ell \cdot \DD{q}{\ell}
			- \sum_{\ell = q+1}^{Q}  \eta_q \cdot \DD{\ell}{q}.
\end{align}
The goal of $\mathcal{D}_{\gamma,\I,M,\beta}$ is to minimize 
\begin{equation}
\textstyle \sum_{q=Q-M+1}^{Q} \xi_q
\end{equation}
subject to the following constraints (in all of them, we omitted the statement
that they hold for all $j \in \{0, \ldots, q\}$):
\begin{align}
	\label{eq:dual-gfqj}
	\DRP{q}{j} + \DG{q}{j} - \DHH{q}{j} 	
		\geq 1
		&&  \text{for all}\; 0 \leq q \leq Q-M \\
	\label{eq:dual-gsqj}
	\DRP{q}{j} - \DB{q}{j}   
		\geq 0 
		&&  \text{for all}\; 0 \leq q \leq Q-M \\
	\label{eq:dual-gfqj2}
	\DRP{q}{j} + \DG{q}{j} - \DHH{q}{j}  + \Dxi{q} 
		\geq 1 
		&&  \text{for all}\; Q-M+1 \leq q \leq Q \\
	\label{eq:dual-gsqj2}
	\DRP{q}{j} - \DB{q}{j}  + \Dxi{q} 	
		\geq 0 
		&&  \text{for all}\; Q-M+1 \leq q \leq Q \\
	\label{eq:dual-wfqj}
\textstyle	\DST{q}{j} 
	+ \eta_{j-1} \cdot \DHH{q}{j} -\eta_j \cdot  \DG{q}{j} 
	+ \DE{L_q} - \sum_{\ell \in S(q)} \DC{\ell} 
		\geq \eta_q 
		&&  \text{for all}\; 0 \leq q \leq Q-M \\
	\label{eq:dual-wsqj}
	\DST{q}{j} + \eta_{j-1} \cdot \DB{q}{j} + \DC{q} 	
		\geq 0
		&&  \text{for all}\; 0 \leq q \leq Q-M \\
	\label{eq:dual-wfqj2}
	\DST{q}{j} -\eta_j \cdot \DG{q}{j} + \eta_{j-1} \cdot \DHH{q}{j}
		\geq \eta_q
		&&  \text{for all}\; Q-M+1 \leq q \leq Q \\
	\label{eq:dual-wsqj2}
	\DST{q}{j} + \eta_{j-1} \cdot \DB{q}{j} - \DE{q}
		\geq 0 
		&&  \text{for all}\; Q-M+1 \leq q \leq Q 
\end{align}
and non-negativity of all variables.

\begin{lemma}
\label{lem:feasibledual3}
For any $\gamma$, any input $\I$ for $\gamma$-resettable scheduling, 
any positive integer~$M$, and any $\beta \in (0,1/M]$, 
there exists a feasible solution to $\mathcal{D}_{\gamma,\I,M,\beta}$ 
of value at most $M + \sum_{j=1}^M (2+\gamma)^{j/M}$.
\end{lemma}

We defer the proof to the next subsection, first arguing how it implies 
the main theorem of the paper (the competitive ratio of \Disc).

\begin{proof}[Proof of \autoref{thm:main}]
Fix any $\gamma$, and consider algorithm $\Disc(\gamma, M, \beta)$ for any
positive integer $M$, and any $\beta \in (0,1/M]$. Fix any input $\I$ to the
$\gamma$-resettable scheduling problem. Let $P^*_{\gamma,\I,M,\beta}$ be the
value of an optimal solution to $\mathcal{P}_{\gamma,\I,M,\beta}$. By weak
duality and \autoref{lem:feasibledual3}, $P_{\gamma,\I,M,\beta} \leq M
+ \sum_{j=1}^M (2+\gamma)^{j/M}$. Hence, by \autoref{lem:comp-ratio-lp},
$\E[\cost_\Disc(\I)] / \cost_\OPT(\I) \leq P^*_{\gamma,\I,M,\beta} / M \leq 1 +
(1/M) \cdot \sum_{j=1}^M (2+\gamma)^{j/M}$, as desired.
\end{proof}


\subsection{Proof of \autoref{lem:feasibledual3}}

Let \[
	\textstyle \Delta_k = \sum_{i=0}^k \alpham^i = \left(\alpham^{k+1} - 1 \right) 
/ \left( \alpham-1 \right). 
\]
In particular $\Delta_{-1} = 0$.
We choose the following values of the dual variables:
\[
	\Dxi{q} = 1 + \alpham^{q-Q+M}
	\quad \text{for}\; Q-M+1 \leq q \leq Q , 
\]
\begin{align*}
	\DB{q}{j} & = 
	\begin{cases}
		\Dxi{q} & \text{for}\; Q - M+1 \leq j \leq q \leq Q , \\
		\delta \cdot \Delta_{M-1}  & \text{for}\; 0 \leq j \leq Q-M \text{ and } q = j , \\ 
	 	1 & \text{for}\; 0 \leq j \leq Q-M \text{ and } q \in \{j+1, \ldots,j+M\}, \\
		0 & \text{otherwise} ,
	\end{cases} \\
	\DG{q}{j} & = 
	\begin{cases} 
		\Delta_{q-Q+M-1} - \Delta_{q-j} & \text{for}\; Q - M+1 \leq j \leq q \leq Q , \\
		\Delta_{q-j-M-1} & \text{for}\; j \leq q - M \\
		0 & \text{otherwise} ,
	\end{cases}
\end{align*}
\begin{align*}
	\DC{q} &= \eta_{q- M -1} \cdot \left(\alpham^{M + 1} + 1 \right) \cdot \left(\alpham^{M} - 1 \right)
		&&& \text{for}\; 0 \leq q \leq Q-M , \\
	 \DE{q} &= \eta_{q-M-1} \cdot \left(\alpham^{M+1} + 1\right)
		&&& \text{for}\; Q-M+1 \leq q \leq Q , \\
	 \DD{q}{j} &= \DB{q,}{j+1} - \DB{q}{j}
		&&& \text{for}\; 0 \leq j < q \leq Q , \\
	 \DHH{q}{j} &= \DB{q}{j} + \DG{q}{j} - 1
		&&& \text{for}\; 0 \leq j \leq q \leq Q .
\end{align*}

The values of $\DB{q}{j}$ and $\DG{q}{j}$ (for $0 \leq j \leq q \leq Q)$ are depicted in 
\autoref{fig:dual-b-g} for an easier reference. We will extensively use the property 
that $\eta_i \cdot \alpham^j = \eta_{i+j}$ for any $i$ and $j$.

\begin{figure}
\centering
\begin{subfigure}{.45\textwidth}
  \centering
  \includegraphics[width=0.85\textwidth]{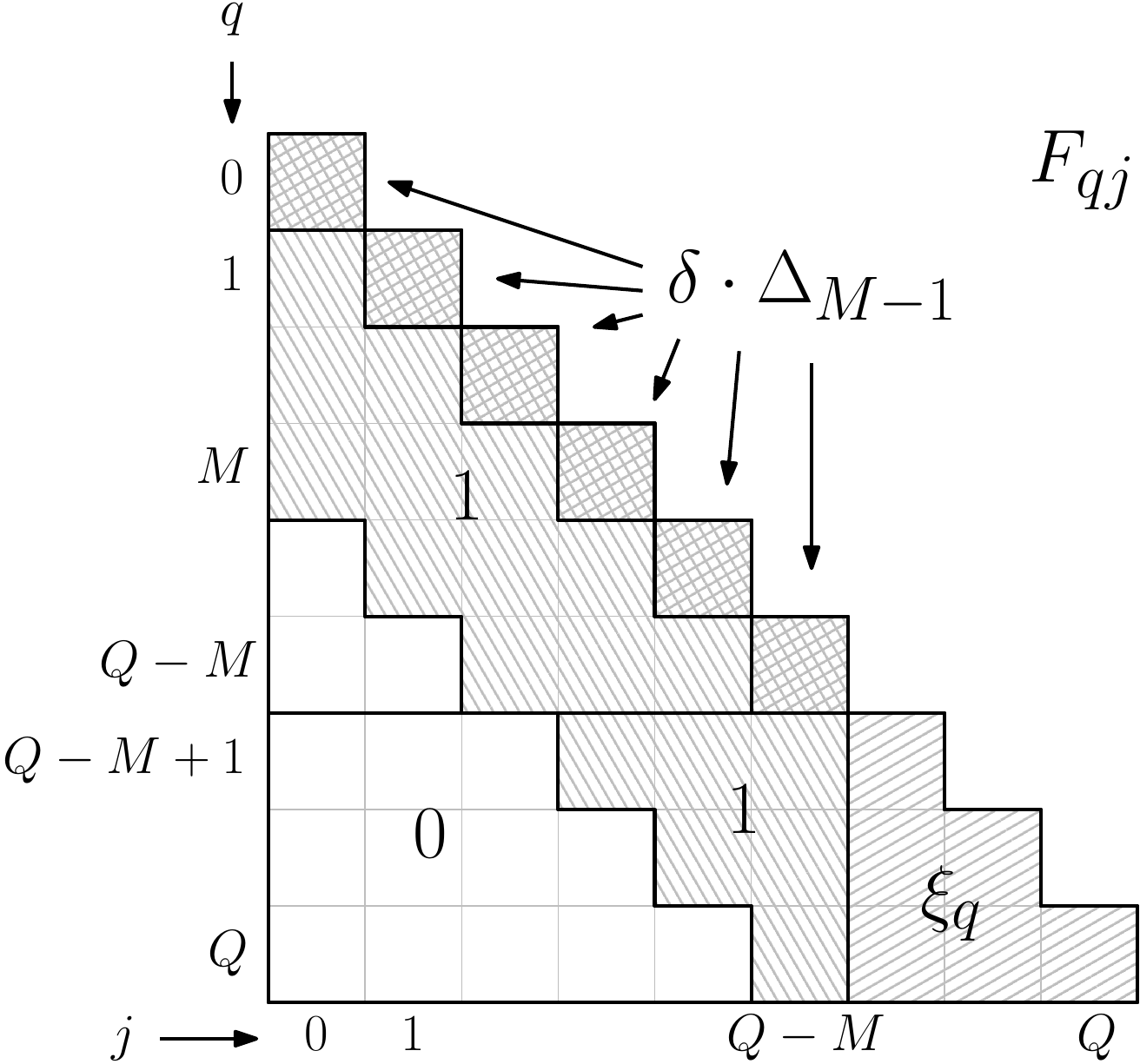}
\end{subfigure}%
\begin{subfigure}{.5\textwidth}
  \centering
  \includegraphics[width=0.85\textwidth]{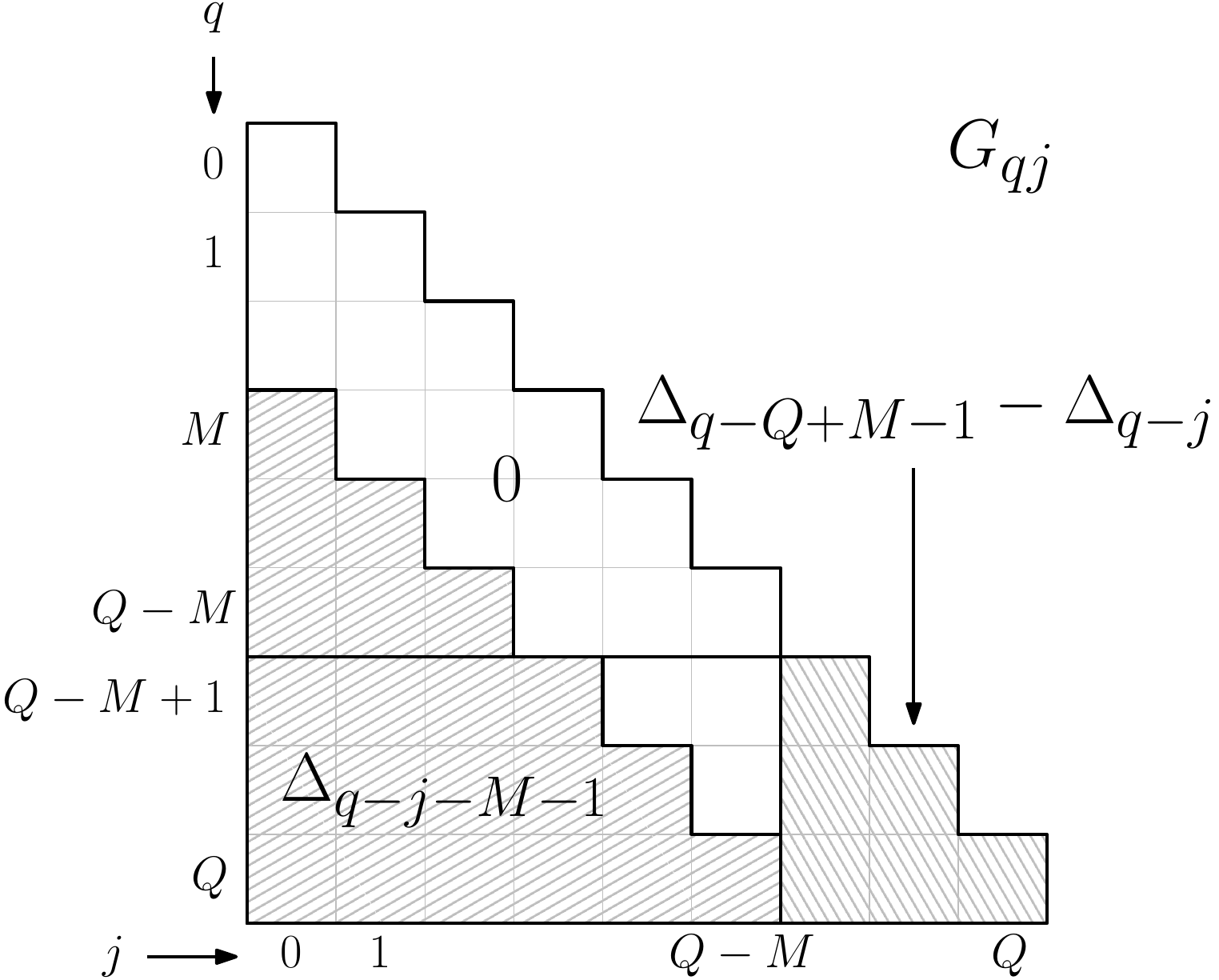}
\end{subfigure}
\caption{Visual presentation of values assigned to dual variables $\DB{q}{j}$ (left)
and $\DG{q}{j}$ (right) for $M = 3$ and $Q = 8$.}
\label{fig:dual-b-g}
\end{figure}

\subparagraph{Objective Value.}

With the above assignment of dual variables the objective value of
$\mathcal{D}_{\gamma,\I,M,\beta}$ is equal to $\sum_{q=Q-M+1}^{Q} \xi_q = M +
\sum_{j=1}^M \alpham^{j} = M + \sum_{j=1}^M (2+\gamma)^{j/M}$ as desired.

\subparagraph{Non-negativity of Variables.}
Variables $\Dxi{q}, \DE{q}, \DC{q}, \DB{q}{j}$ and $\DG{q}{j}$ are trivially
non-negative (for those $q$ and $j$ for which they are defined). The non-negativity
of $\DD{q}{j} = \DB{q,}{j+1} - \DB{q}{j}$ follows as
$\DB{q}{j}$ is a non-decreasing function of its second argument
(cf.~\autoref{fig:dual-b-g}).

Finally, for showing non-negativity of variable $\DHH{q}{j}$, we consider two cases. 
If $j \geq q - M$, then $\DB{q}{j} \geq 1$. Otherwise, $j \leq q -M -1$, and then 
$\DG{q}{j} = \Delta_{q-j-M-1} \geq 1$. Thus, in either case $\DHH{q}{j} 
= \DB{q}{j} + \DG{q}{j} - 1 \geq 0$.

\subparagraph{Helper Bounds.}

It remains to show that the given values of dual variables 
satisfy all constraints \eqref{eq:dual-gfqj}--\eqref{eq:dual-wsqj2}
of the dual program $\mathcal{D}_{\gamma,\I,M,\beta}$.
We define a few helper notions and identities that are used throughout the proof of dual feasibility.
For any $q \in [Q+1]$, let 
\begin{align*}
	\textstyle R_q = \sum_{\ell = q+1}^{Q} \DD{\ell}{q}
	= \sum_{\ell=q+1}^{Q} \left( \DB{\ell,}{q+1} - \DB{\ell}{q} \right) .
\end{align*}

\begin{lemma}
$R_q = \delta \cdot \Delta_{M-1}$ for  $q \leq Q-M$ and $R_q = 0$ otherwise.
\end{lemma}

\begin{proof}
We consider three cases.
\begin{enumerate}
\item $q \in \{0, \dots, Q-M-1\}$. Then, 
$R_q = \DB{q+1,}{q+1} 
	+ \sum_{\ell=q+1}^{Q} \left( \DB{\ell+2,}{\ell+1} - \DB{\ell+1,}{\ell} \right) 
	- \DB{Q}{q} = \delta \cdot \Delta_{M-1} + \sum_{\ell=q+1}^Q 0 - 0 = \delta \cdot \Delta_{M-1}$.
\item $q = Q-M$. Then,
$R_q = \sum_{\ell=Q-M+1}^Q (\xi_\ell - 1) = \sum_{j=1}^M \alpham^j = \delta \cdot \Delta_{M-1}$.

\item $q \in \{Q-M+1,\dots,Q\}$. Then, 
$R_q = \sum_{\ell=q+1}^{Q} (\xi_\ell - \xi_\ell) = 0$.
\qedhere
\end{enumerate}
\end{proof}

Next, we investigate the values of $\DST{q}{j}$ for different $q$ and $j$. 
Using its definition (cf.~\eqref{eq:UV_def}), 
\begin{equation}
\label{eq:dst}
\textstyle	\DST{q}{j} 
	= \sum_{\ell = j}^{q-1}  \eta_\ell \cdot \DD{q}{\ell}
			- \sum_{\ell = q+1}^{Q}  \eta_q \cdot \DD{\ell}{q}
	= \sum_{\ell = j}^{q-1}  \eta_\ell \cdot 
		\left( \DB{q,}{\ell+1} - \DB{q}{\ell} \right)
			- \eta_q \cdot R_q.
\end{equation}
Additionally, using $\DHH{q}{j} = \DB{q}{j} + \DG{q}{j} - 1$, we obtain
\begin{equation}
\label{eq:GH_relation}
\eta_j \cdot \DG{q}{j} - \eta_{j-1} \cdot \DHH{q}{j} 
	= (\eta_{j} - \eta_{j-1}) \cdot \DG{q}{j} + \eta_{j-1} - \eta_{j-1} \cdot \DB{q}{j} .
\end{equation}
Using the chosen values of $\DG{q}{j}$, we observe that
\begin{equation}
\label{eq:H_and_G_better}
	(\eta_{j} - \eta_{j-1}) \cdot \DG{q}{j} = 
	\begin{cases}
		\eta_{q+j-Q+M-1} - \eta_{q} & \text{for}\; Q-M+1 \leq j \leq q, \\
		\eta_{q-M-1} - \eta_{j-1} & \text{for}\; j \leq q - M -1, \\
		0 & \text{otherwise}.
	\end{cases}
\end{equation}
Furthermore, in all the cases, it can be verified that
\begin{equation}
\label{eq:G_relation}
	(\eta_{j} - \eta_{j-1}) \cdot \DG{q}{j} + \eta_{j-1} - \eta_{q-M-1} \geq 0.
\end{equation}

\subsubsection{Showing inequalities (\ref{eq:dual-gfqj})--(\ref{eq:dual-gsqj2})}

We prove that relations \eqref{eq:dual-gfqj}--\eqref{eq:dual-gsqj2} hold with equality. 
In fact, it suffices to show \eqref{eq:dual-gsqj} and~\eqref{eq:dual-gsqj2}:
inequalities \eqref{eq:dual-gfqj} and \eqref{eq:dual-gfqj2} follow immediately 
as we chose $\DHH{q}{j} = \DB{q}{j} + \DG{q}{j} - 1$.
Using the definition of $\DRP{q}{j}$ (cf.~\eqref{eq:UV_def}), we obtain
\[ 
	\textstyle \DRP{q}{j} = \sum_{\ell = q+1}^{Q} \DD{\ell}{q}
			- \sum_{\ell = j}^{q-1}  \DD{q}{\ell}
			= R_q - \sum_{\ell=j}^{q-1} \left( \DB{q,}{\ell+1} - \DB{q}{\ell} \right) 
			= R_q - \DB{q}{q} + \DB{q}{j} .
\]
Now, we observe that for $q \leq Q-M$, it holds that $R_q - \DB{q}{q} = \delta \cdot \Delta_{M-1} - \delta \cdot \Delta_{M-1} = 0$,
and thus $\DRP{q}{j} - \DB{q}{j} = 0$, which implies \eqref{eq:dual-gsqj}. 
On the other hand, for $q > Q-M$, it holds that $R_q - \DB{q}{q} = 0 - \xi_q$, 
and hence $\DRP{q}{j} - \DB{q}{j} + \xi_q = 0$, which implies \eqref{eq:dual-gsqj2}.

\subsubsection{Showing inequalities (\ref{eq:dual-wfqj})--(\ref{eq:dual-wsqj})}

Within this part, we assume $q \leq Q-M$. 
We start with evaluating some terms that are present in~\eqref{eq:dual-wfqj} and \eqref{eq:dual-wsqj}.
First, we observe that 
\begin{align}
\label{eq:C_q}
	\DC{q} = \eta_{q- M -1} \cdot \left(\alpham^{M + 1} + 1 \right) \cdot \left(\alpham^{M} - 1 \right) 
	 = \eta_{q+M} - \eta_q + \eta_{q-1} - \eta_{q-M-1}.
\end{align}
Second, we compute the term $\DE{L_q} - \sum_{\ell \in S(q)} \DC{\ell}$.
Recall that $S(q) = \{q+M, q+2\cdot M, \dots, L_q - M\}$. Thus,
\begin{align}
\nonumber
\textstyle	\DE{L_q} - \sum_{\ell \in S(q)} \DC{\ell}
	& = \textstyle \left(\alpham^{M+1} + 1 \right) \cdot \left[ 
			\eta(L_q-M-1)  - 
			(\alpham^{M}-1) \cdot \eta(-M-1) \cdot \sum_{\ell \in S(q)} \alpham^\ell 
		\right] \\
\nonumber
	& = (\alpham^{M+1}+1) \cdot \left[ \eta(L_q-M-1) - \eta(-M-1) \cdot \left(\alpham^{L_q} - \alpham^{q+M} \right) \right] \\
	\label{eq:E_Csum}
	& = (\alpham^{M+1}+1) \cdot \eta_{q-1} 
	 = \eta_{q+M} + \eta_{q-1} .
\end{align}

\begin{lemma}
\label{lem:dual-tech}
Fix any $0 \leq j \leq q \leq Q-M$. Then,
\[
	\DST{q}{j} 
	= \eta_q - \eta_{q-1} - \eta_{q+M} 
	- \eta_{j-1} \cdot \DB{q}{j} + (\eta_{j} - \eta_{j-1}) \cdot \DG{q}{j} + \eta_{j-1}.
\]
\end{lemma}

\begin{proof}
By the definition, $\Delta_{M-1} = \sum_{i=0}^{M-1} \alpham^i$, 
and therefore $(\eta_{q-1}-\eta_{q}) \cdot \delta \cdot \Delta_{M-1} = \eta_q - \eta_{q+M}$.
Thus, it suffices to show the following relation
\[
	\DST{q}{j} 
	= \eta_{q-1} \cdot (\delta \cdot \Delta_{M-1} - 1) 
		- \eta_{q} \cdot \delta \cdot \Delta_{M-1} 
		- \eta_{j-1} \cdot \DB{q}{j} + (\eta_{j} - \eta_{j-1}) \cdot \DG{q}{j} + \eta_{j-1}.
\]
To evaluate $\DST{q}{j}$ using \eqref{eq:dst}, it is useful to trace values
$\DB{q}{j}, \DB{q,}{j+1}, \dots, \DB{q}{q}$ (cf.~\autoref{fig:dual-b-g}), noting
that only the increases of these values contribute to $\DST{q}{j}$.
We also note that for $q \leq Q-M$, possible increases are from $0$ to $1$
(between $\DB{q,}{q-M-1}$ and $\DB{q,}{q-M}$) and from $1$ to $\delta \cdot
\Delta_{M-1}$ (between $\DB{q,}{q-1}$ and $\DB{q}{q}$). We consider three cases,
using $R_q = \alpham \cdot \Delta_{M-1}$ below.

\begin{enumerate}
\item $j \leq q-M-1$. 
Then, $\DB{q}{j} = 0$ and 
\begin{align*}
	\DST{q}{j} 
	& = \eta_{q-1} \cdot \left( \DB{q}{q} - \DB{q,}{q-1} \right)
		+ \eta_{q-M-1} \cdot \left( \DB{q,}{q-M} - \DB{q,}{q-M-1} \right)
			- \eta_q \cdot R_q \\
	& = \eta_{q-1} \cdot \left( \delta \cdot \Delta_{M-1} - 1 \right)
		- \eta_q \cdot  \delta \cdot \Delta_{M-1}
		+ \eta_{q-M-1} - \eta_{j-1} + \eta_{j-1} - \eta_{j-1} \cdot \DB{q}{j}. 
\end{align*}
The lemma follows as 
$(\eta_{j} - \eta_{j-1}) \cdot \DG{q}{j} = \eta_{q-M-1} - \eta_{j-1}$ (see~\eqref{eq:H_and_G_better}).

\item $j \in \{ q-M, \ldots, q-1 \}$. 
Then $\DB{q}{j} = 1$, and 
\begin{align*}
	\DST{q}{j} 
	& = \eta_{q-1} \cdot \left( \DB{q}{q} - \DB{q,}{q-1} \right)
			- \eta_q \cdot R_q \\
	& = \eta_{q-1} \cdot \left( \delta \cdot \Delta_{M-1} - 1 \right)
		- \eta_q \cdot \delta \cdot \Delta_{M-1} \\
	& = \eta_{q-1} \cdot \left( \delta \cdot \Delta_{M-1} - 1 \right)
		- \eta_q \cdot \delta \cdot \Delta_{M-1} - \eta_{j-1} \cdot \DB{q}{j} + \eta_{j-1}. 
\end{align*}
The lemma follows as $(\eta_{j} - \eta_{j-1}) \cdot \DG{q}{j} = 0$ (see~\eqref{eq:H_and_G_better}).

\item $j = q$. 
Then $\DB{q}{j} = \delta \cdot \Delta_{M-1}$, and thus
\begin{align*}
	\DST{q}{j} 
	&=  - \eta_q \cdot R_q \\
	&=  \eta_{q-1} \cdot \delta \cdot \Delta_{M-1} 
		- \eta_q \cdot \delta \cdot \Delta_{M-1} 
		- \eta_{j-1} \cdot \DB{q}{j} \\
	&=  \eta_{q-1} \cdot \left( \delta \cdot \Delta_{M-1} - 1 \right)
		- \eta_q \cdot \delta \cdot \Delta_{M-1} 
		- \eta_{j-1} \cdot \DB{q}{j} + \eta_{j-1} .
\end{align*}
As in the previous case, the lemma follows as $(\eta_{j} - \eta_{j-1}) \cdot \DG{q}{j} = 0$.
\qedhere
\end{enumerate}

\end{proof}

\subparagraph{Showing Inequality (\ref{eq:dual-wfqj}).}

We show that \eqref{eq:dual-wfqj} holds with equality.
Using \autoref{lem:dual-tech}, \eqref{eq:E_Csum}, and \eqref{eq:GH_relation} yields
\begin{align*}
	\DST{q}{j} 
		& \textstyle  + \eta_{j-1} \cdot \DHH{q}{j} -\eta_j \cdot  \DG{q}{j} + 
	\DE{L_q} - \sum_{\ell \in S(q)} \DC{\ell} \\
		& = \eta_q - \eta_{q-1} - \eta_{q+M} 
		- \eta_{j-1} \cdot \DB{q}{j} + (\eta_{j} - \eta_{j-1}) \cdot \DG{q}{j} + \eta_{j-1} \\
		& \quad - (\eta_{j} - \eta_{j-1}) \cdot \DG{q}{j} - \eta_{j-1} + \eta_{j-1} \cdot \DB{q}{j} 
		 + \eta_{q+M} + \eta_{q-1} 
		 = \eta_q .
\end{align*}

\subparagraph{Showing Inequality (\ref{eq:dual-wsqj}).}

Using \autoref{lem:dual-tech}, \eqref{eq:E_Csum}, and \eqref{eq:GH_relation} yields
\begin{align*}
	\DST{q}{j} & + \eta_{j-1} \cdot \DB{q}{j} + \DC{q} 	\\
		& = \eta_q - \eta_{q-1} - \eta_{q+M} 
		- \eta_{j-1} \cdot \DB{q}{j} + (\eta_{j} - \eta_{j-1}) \cdot \DG{q}{j} + \eta_{j-1} \\
		& \quad 
		+ \eta_{j-1} \cdot \DB{q}{j} + \eta_{q+M} - \eta_q + \eta_{q-1} - \eta_{q-M-1} \\
		& = (\eta_{j} - \eta_{j-1}) \cdot \DG{q}{j} + \eta_{j-1} - \eta_{q-M-1} 
		\geq 0.
\end{align*}
where the last inequality follows by \eqref{eq:G_relation}.


\subsubsection{Showing inequalities (\ref{eq:dual-wfqj2})--(\ref{eq:dual-wsqj2})}

Within this part, we assume that $q \geq Q-M+1$. 

\begin{lemma}
\label{lem:dual-tech-2}
Fix any $q \geq Q-M+1$ and $0 \leq j \leq q$. Then,
\[
	\DST{q}{j} 
	= \eta_q + (\eta_{j} - \eta_{j-1}) \cdot \DG{q}{j}- \eta_{j-1} \cdot \DB{q}{j} + \eta_{j-1}.
\]
\end{lemma}

\begin{proof}
As $g \geq Q-M+1$, it holds that $R_q = 0$, and thus \eqref{eq:dst} reduces to 
\[
\textstyle	\DST{q}{j} 
	= \sum_{\ell = j}^{q-1}  \eta_\ell \cdot 
		\left( \DB{q,}{\ell+1} - \DB{q}{\ell} \right).
\]
As in the proof of \autoref{lem:dual-tech}, 
to further evaluate $\DST{q}{j}$, it is useful to trace values $\DB{q}{j}, \DB{q,}{j+1}, \dots,$ $ \DB{q}{q}$
(cf.~\autoref{fig:dual-b-g}), where the increases of these values contribute to $\DST{q}{j}$. 
We also note that for $q \geq Q-M+1$, the possible increases are from $0$ to $1$ 
(between $\DB{q,}{q-M-1}$ and $\DB{q,}{q-M}$) and 
from $1$ to $\Dxi{q}$ (between $\DB{q,}{Q-M}$ and $\DB{q,}{Q-M+1}$).
We consider three cases.
\begin{enumerate}
\item $j \leq q-M-1$. 
Then $\DB{q}{j} = 0$, and 
\begin{align*}
	\DST{q}{j} 
	& = \eta_{Q-M} \cdot \left( \DB{q}{q} - \DB{q,}{q-1} \right)
		+ \eta_{q-M-1} \cdot \left( \DB{q,}{q-M} - \DB{q,}{q-M-1} \right) \\
	& = \eta_{Q-M} \cdot \left( \Dxi{q} - 1 \right)
		+ \eta_{q-M-1}  \\
	& = \eta_{q} + \eta_{q-M-1} -\eta_{j-1} + \eta_{j-1} - \eta_{j-1} \cdot \DB{q}{j}. 
\end{align*}
The lemma follows as 
$(\eta_{j} - \eta_{j-1}) \cdot \DG{q}{j} = \eta_{q-M-1} - \eta_{j-1}$ (see~\eqref{eq:H_and_G_better}).

\item $j \in \{ q-M, \ldots, Q-M \}$. 
Then $\DB{q}{j} = 1$, and 
\begin{align*}
	\DST{q}{j} 
	& = \eta_{Q-M} \cdot \left( \DB{q}{q} - \DB{q,}{q-1} \right) \\
	& = \eta_{Q-M} \cdot \left( \Dxi{q} - 1 \right) \\
	& = \eta_{q} - \eta_{j-1} \cdot \DB{q}{j} + \eta_{j-1}.
\end{align*}
The lemma follows as 
$(\eta_{j} - \eta_{j-1}) \cdot \DG{q}{j} = 0$ (see~\eqref{eq:H_and_G_better}).

\item $j \in \{ q-M, \ldots, Q-M \}$. 
Then $\DB{q}{j} = \Dxi{q} = 1 + \delta^{q-Q+M}$, and 
\begin{align*}
	\DST{q}{j} 
	= 0 
	& = \eta_{q} + \eta_{q+j-Q+M-1} - \eta_{q} - \eta_{j-1} \cdot (1 + \delta^{q-Q+M}) + \eta_{j-1} \\
	& = \eta_{q} + \eta_{q+j-Q+M-1} - \eta_{q} - \eta_{j-1} \cdot \DB{q}{j} + \eta_{j-1} .
\end{align*}
The lemma follows as 
$(\eta_{j} - \eta_{j-1}) \cdot \DG{q}{j} = \eta_{q+j-Q+M-1} - \eta_{q}$ (see~\eqref{eq:H_and_G_better}).
\qedhere
\end{enumerate}
\end{proof}

\subparagraph{Showing Inequality (\ref{eq:dual-wfqj2}).}

We show that \eqref{eq:dual-wfqj2} holds with equality.
Using \autoref{lem:dual-tech-2} and~\eqref{eq:GH_relation}, we obtain
\begin{align*}
	\DST{q}{j} 
		& \textstyle  + \eta_{j-1} \cdot \DHH{q}{j} -\eta_j \cdot  \DG{q}{j}  \\
		&= \eta_q + (\eta_{j} - \eta_{j-1}) \cdot \DG{q}{j}- \eta_{j-1} \cdot \DB{q}{j} + \eta_{j-1}
			- (\eta_{j} - \eta_{j-1}) \cdot \DG{q}{j} - \eta_{j-1} + \eta_{j-1} \cdot \DB{q}{j} \\
		&= \eta_q.
\end{align*}

\subparagraph{Showing Inequality (\ref{eq:dual-wsqj2}).}

Using \autoref{lem:dual-tech-2}, \eqref{eq:GH_relation}, and the definition of $\DE{q}$, we obtain
\begin{align*}
	\DST{q}{j} 
		& \textstyle  + \eta_{j-1} \cdot \DB{q}{j} - \DE{q}  \\
		&= \eta_q + (\eta_{j} - \eta_{j-1}) \cdot \DG{q}{j}- \eta_{j-1} \cdot \DB{q}{j} + \eta_{j-1}
			+ \eta_{j-1} \cdot \DB{q}{j} - \eta_{q} - \eta_{q-M-1} \\
		&= (\eta_{j} - \eta_{j-1}) \cdot \DG{q}{j} + \eta_{j-1} - \eta_{q-M-1} \geq 0.
\end{align*}
where the last inequality follows by \eqref{eq:G_relation}.


\section{Tightness of the Analysis}

The analysis of our algorithms is tight as proven below. For the deterministic
one, we additionally show that choosing $\omega$ different from $0$ does not
help.

\begin{theorem}
For any $\gamma$, there are $\gamma$-resettable scheduling problems, such that
for any $\omega \in (-1,0]$, the competitive ratio of $\Mimic(\gamma,\omega)$ is at least $3+\gamma$.
\end{theorem}

\begin{proof}
We fix a small $\varepsilon > 0$ and let $\alpha = 2 + \gamma$. The input $\I$
contains two jobs: the first one of weight $\varepsilon$ that arrives at time
$1$, and second one of weight $1$ that arrives at time $\alpha^{1+\omega} +
\varepsilon$. We assume that there exists a schedule $S_1$ that serves the first
job at the time of its arrival and a schedule $S_2$ that serves both jobs at the
times of their arrivals. Therefore, $\cost_{\OPT}(\I) = \varepsilon \cdot 1 + 1
\cdot (\alpha^{1+\omega} + \varepsilon) = \alpha^{1+\omega} + 2 \cdot
\varepsilon$.

For analyzing the cost of $\Mimic$, note that at at time $1$, $\Mimic$ observes
the first job and learns the value of $\min(\I) = 1$. This is the sole purpose
of the first job: setting $\min(\I) = 1$ makes the algorithm miss the
opportunity to serve the second job early. At time $\tau_1 = \alpha^{1+\omega}$,
$\Mimic$ executes the $\tau_1$-schedule $S'_1$, which is schedule $S_1$
prolonged trivially to length $\tau_1$. Next, at time $\tau_2 =
\alpha^{2+\omega}$, $\Mimic$ executes the $\tau_2$-schedule $S'_2$, which is
schedule $S_2$ prolonged trivially to length $\tau_2$. This way it completes the
second job at time $\tau_2 + (\alpha^{1+\omega} + \varepsilon)$, and thus
$\cost_\Mimic(\I) \geq \tau_2 + \alpha^{1+\omega} + \varepsilon =
\alpha^{1+\omega} \cdot (1+\alpha) + \varepsilon$. By taking appropriately small
$\varepsilon > 0$, the ratio between $\cost_\Mimic(\I)$ and $\cost_\OPT(\I)$
becomes arbitrarily close to $1+\alpha = 3+\gamma$. 
\end{proof}

\begin{theorem}
For any $\gamma$, there are $\gamma$-resettable scheduling problems, such that
the competitive ratio of a randomized algorithm that runs $\Mimic(\gamma,\omega)$ 
with a random $\omega \in (-1,0]$ is at least $1+(1+\gamma)/\ln(2+\gamma)$.
\end{theorem}

\begin{proof}
Let $\alpha = 2+\gamma$. The input $\I$ contains a single job of weight $1$
arriving at time $1$. We also assume that for any $\tau \geq 1$, there exists a
$\tau$-schedule $S_{\tau}$ that completes this job at time~$1$. Clearly,
$\cost_{\OPT}(\I) = 1 \cdot 1 = 1$.

At time $1$, $\Mimic$ observes the only job of $\I$ and learns that
$\min(\I) = 1$. Its sets $\tau_1 = \alpha^{1+\omega}$ and at time $\tau_1$ it
executes schedule $S_{\tau_1}$, thus completing the job at time $\tau_1 + 1$.
Therefore, $\cost_\Mimic(\I) = \int_{-1}^0 \tau_1+1 \;d\omega = \int_{-1}^0
\alpha^{1+\omega} + 1 \;d\omega = 1 + (\alpha-1) / \ln \alpha  
= 1+(1+\gamma)/\ln(2+\gamma)$. This implies the desired lower bound.
\end{proof}


\section{Flaw in the Randomized Lower Bound for DARP}
\label{sec:flaw}

The authors of~\cite{FiKrWe09} claim a lower bound of $3$ for randomized
$k$-DARP (for any $k \geq 1$), see Theorem 4 of~\cite{FiKrWe09}. Below we show a
flaw in their argument.

The construction given in the proof of their Theorem 4 uses Yao min-max principle
and is parameterized with a few variables, in particular with an integer $m$ and 
with a real number $v \in [0,1]$. Towards the end of the proof, they show that 
the competitive ratio of any randomized algorithm for the $k$-DARP problem is at least 
\[
    L_{m,v} = \frac{3m - 4km - 4km^2 + v^{\frac{m+1}{2-m}} \cdot (3 + (4km^2 + 4km + 6m + 6) \cdot v)}
        {-4 -m - 2km - 2km^2 + v^{\frac{m+1}{2-m}} \cdot (3 + (2km^2 + 2km + 4m + 4) \cdot v)}
\]
and they claim that there exists $v$, such that 
$L_{m,v} = 3$ when $m$ tends to infinity. However, for any fixed $k$ and any $v$ (also being a function
of $m$), by dividing numerator and denominator by $m^2$, we obtain that
\[
    \lim_{m \to \infty} L_{m,v} = 
    \frac{- 4k + v^{\frac{m+1}{2-m}} \cdot 4k \cdot v}
    {- 2k + v^{\frac{m+1}{2-m}} \cdot 2k \cdot v}
    = 2.
\]
That is, the proven lower bound is $2$ instead of $3$.


\bibliographystyle{plainurl}
\bibliography{trp}

\end{document}